\documentclass[a4paper,UKenglish]{lipics-v2016}
\RequirePackage{fix-cm}

\usepackage{hyperref}
\hypersetup{%
   breaklinks,%
   colorlinks=true,%
   linkcolor=[rgb]{0.45,0.0,0.0},%
   urlcolor=[rgb]{0.05,0.390,0.0},%
   citecolor=[rgb]{0,0,0.45}                                          
}%
\usepackage{amsmath,amssymb,latexsym,epsfig,amsthm}   
\usepackage{authblk}
\usepackage[nameinlink]{cleveref} 
\usepackage{microtype}
\usepackage{lmodern} 

\usepackage{chngcntr}
\usepackage{apptools}

\usepackage{paralist}

\usepackage{amssymb,latexsym,amsthm}
\usepackage{adjustbox}
\usepackage{algorithm}

\usepackage{enumerate}
\usepackage{appendix}
\usepackage{amsmath}
\usepackage{enumitem} 
\usepackage[numbers]{natbib}

\usepackage{etoolbox}
\makeatletter
    \pretocmd{\NAT@citexnum}{\@ifnum{\NAT@ctype>\z@}{\let\NAT@hyper@\relax}{}}{}{}
\makeatother

\newcommand{\CCite}[1]{\citeauthor{#1}~\cite{#1}}

\usepackage{xspace} 
\usepackage{color}%

\usepackage{pgf,tikz}

\usepackage{thmtools, thm-restate} 
\usepackage[nameinlink]{cleveref} 

\usepackage[noend]{algorithmic}
\Crefname{ALC@unique}{Line}{Lines} 

\crefname{constr}{constraint}{constraints}
\crefname{ineq}{inequality}{inequalities}
\creflabelformat{ineq}{#2{\upshape(#1)}#3} 
\crefname{claim}{claim}{claims}           
\crefname{defn}{definition}{definitions}  
\crefname{obs}{Observation}{Observations}
\crefname{lem}{Lemma}{Lemmas}
\crefname{cor}{Corollary}{Corollaries}

\newlist{clenum}{enumerate}{1} 
\setlist[clenum]{label=(\alph*),ref=\textup{\theclaim~(\alph*)}}
\crefalias{clenumi}{claim} 
\newlist{thenum}{enumerate}{1} 
\setlist[thenum]{label=(\alph*),ref=\textup{\thetheorem~(\alph*)}}
\crefalias{thenumi}{theorem} 

\newtheorem{lem}{Lemma}

\newtheorem{claim}{Claim}[section]
\newtheorem{defn}{Definition}[section]
\newtheorem{obs}[defn]{Observation}
\newtheorem{cor}{Corollary}

\newcommand{\lr}[1]{\bigg(#1\bigg)}

\newcommand{\cb}{\mathcal{B}}
\newcommand{\cf}{\mathcal{F}}
\newcommand{\ch}{\mathcal{H}}
\newcommand{\chp}{\mathcal{H}_1}
\newcommand{\cl}{\mathcal{L}}

\newcommand{\real}{\mathbb{R}}
\newcommand{\ball}{\mathcal{B}}
\newcommand{\cost}{\mathrm{cost}}
\newcommand{\sol}{\sigma}
\newcommand{\solbar}{\overline{\sigma}}
\newcommand{\heavybb}{\mathcal{H}_1}
\newcommand{\FLightBb}{\mathcal{F}}
\newcommand{\OLightBb}{\mathcal{O}}
\newcommand{\clusterbb}{\mathcal{C}}
\newcommand{\clst}{\text{ cluster}}
\newcommand{\yac}{\widetilde{y}}
\newcommand{\AvCap}{\text{AvCap}}

\begin{document}


\title{Capacitated Covering Problems in Geometric Spaces\footnote{This material is based upon work supported by
the National Science Foundation under Grant CCF-1615845}}

\author{Sayan Bandyapadhyay
} 
\author{Santanu Bhowmick
} 
\author{Tanmay Inamdar
} 
\author{Kasturi Varadarajan
}
\affil{
  Department of Computer Science, 
  University of Iowa, Iowa City, USA\thanks{\{sayan-bandyapadhyay, santanu-bhowmick, tanmay-inamdar, kasturi-varadarajan\}@uiowa.edu}\\
  }

\date{}
%
  
\authorrunning{S.\,Bandyapadhyay, S.\,Bhowmick, T.\,Inamdar and K.\,Varadarajan} 
%
%
  
\maketitle
\begin{abstract}
In this article, we consider the following capacitated covering problem. We are given a set $P$ of $n$ points and a set $\cb$ of balls from some metric space, and a positive integer $U$ that represents the {\em capacity} of each of the balls in $\cb$. We would like to compute a subset $\cb' \subseteq \cb$ of balls and assign each point in $P$ to some ball in $\cb'$ that contains it, such that the number of points assigned to any ball is at most $U$. The objective function that we would like to minimize is the cardinality of $\cb'$. 

We consider this problem in arbitrary metric spaces as well as Euclidean spaces of constant dimension. In the metric setting, even the uncapacitated version of the problem is hard to approximate to within a logarithmic factor. In the Euclidean setting, the best known approximation guarantee in dimensions $3$ and higher is logarithmic in the number of points. Thus we focus on obtaining ``bi-criteria'' approximations. In particular, we are allowed to expand the balls in our solution by some factor, but optimal solutions do not have that flexibility. Our main result is that allowing constant factor expansion of the input balls suffices to obtain constant approximations for these problems. In fact, in the Euclidean setting, only $(1+\epsilon)$ factor expansion is sufficient for any $\epsilon > 0$, with the approximation factor being a polynomial in $1/\epsilon$. We obtain these results using a unified scheme for rounding the natural LP relaxation; this scheme may be useful for other capacitated covering problems. 
We also complement these bi-criteria approximations by obtaining hardness of approximation results that shed light on our understanding of these problems.

\end{abstract}

\section{Introduction}
\label{sec:intro}
In this paper, we consider the following capacitated covering problem. We are given a set $P$ of $n$ points and a set $\cb$ of balls from some metric space, and a positive integer $U$ that represents the {\em capacity} of each of the balls in $\cb$. We would like to compute a subset $\cb' \subseteq \cb$ of balls and assign each point in $P$ to some ball in $\cb'$ that contains it, such that the number of points assigned to any ball is at most $U$. The objective function that we would like to minimize is the cardinality of $\cb'$. We call this the \textit{Metric Capacitated Covering} (MCC) problem.

An important special case of this problem arises when $U = \infty$, and we refer to this as \textit{Metric Uncapacitated Covering} (MUC). The MUC requires us to cover the points in $P$ using a minimum number of balls from $\cb$, and we can therefore solve it using the efficient greedy algorithm for set cover and obtain an approximation guarantee of $O(\log n)$. The approximation factor of $O(\log n)$ for Set Cover cannot be improved under widely held complexity theoretic
assumptions~\cite{Feige98}. The same is true for the MUC, as 
demonstrated by the following reduction from Set Cover. We take a ball of radius $1$ corresponding to each set, and a point corresponding to each element. If an element is in a set, then the distance between the center of the corresponding ball and the point is $1$. We consider the metric space induced by the centers and the points. It is easy to see that any solution for this instance of MUC directly gives a solution for the
input instance of the general Set Cover, implying that for MUC, it is not
possible to get any approximation guarantee better than the $O(\log n)$ bound
for Set Cover.

The MUC in fixed dimensional Euclidean spaces has been extensively studied. One interesting variant is when the allowed set $\cb$ of balls consists of all {\em unit} balls. \CCite{HochbaumM85} gave a
polynomial time approximation scheme (PTAS) for this using a grid shifting 
strategy. When $\cb$ is an arbitrary finite set of balls, the problem seems to be much harder. An $O(1)$ approximation algorithm in the 2-dimensional Euclidean plane was given by \CCite{BronnimannG1995}. More recently, a PTAS was obatined by \CCite{MustafaR2010}. 
In dimensions $3$ and higher, the best known approximation guarantee is still
$O(\log n)$. Motivated by this, \CCite{Har-PeledL12} gave a PTAS for a bi-criteria version where the algorithm is allowed to expand the input balls by a $(1 + \epsilon)$ factor. Covering with geometric objects other than balls has also been extensively studied; see \cite{Lev-TovP05,ClarksonV07,AronovES2010,VaradarajanWGSC2010, ChanGKS12,GovindarajanRRB2016} for a sample. 

The MCC is a special case of the  \textit{Capacitated Set Cover} (CSC) problem. In the latter problem, we are given a set system $(X, \cf)$ with $n = |X|$ elements and $m = |\cf|$ subsets of $X$. For each set $\cf_i \in \cf$, we are also given an integer $U_i$, which is referred to as its \textit{capacity}. We are required to find a
minimum size subset $\cf' \subseteq \cf$ and assign each element in $X$ to a 
set in $\cf'$ containing it, such that for each set $\cf_i$, the number of
points assigned to $\cf_i$ is at most $U_i$. 
The MCC is obtained as a special case of CSC by setting $X = P$, $\cf = \cb$, and $U_i = U \forall i$.
Set Cover is a special 
case of CSC where the capacity of each set is $\infty$. 

Applications of Set Cover include placement of wireless sensors or
antennas to serve clients, VLSI design, and image processing~\cite{BermanKL12,HochbaumM85}. 
It is natural to consider capacity constraints that appear in many
applications, for instance, an upper bound on the number of clients that can be
served by an antenna. Such constraints lead to the natural formulation of
CSC.
For the CSC problem,~\CCite{Wolsey82} used a
greedy algorithm to give an $O(\log
n)$ approximation. For the special case of vertex cover (where each element in $X$ belongs to exactly two
sets in $\cf$),~\CCite{ChuzhoyN06} presented an algorithm with approximation ratio $3$, which was
subsequently improved to $2$ by~\CCite{GandhiHKKS06}. The generalization where
each element belongs to at most a bounded number $f$ of sets has been studied in
a sequence of works, culminating in~\cite{Kao17,Wong17}. \CCite{BermanKL12} have considered the ``soft'' capacitated version of the CSC problem that allows making multiple copies of input sets. Another closely related problem to the CSC problem is the so-called \emph{Replica Placement} problem. For the graphs of treewidth bounded by $t$, an $O(t)$ approximation algorithm for this problem is presented in \cite{AggarwalCGSST17}. Finally, PTASes for the Capacitated Dominating Set, and Capacitated Vertex Cover problems on the planar graphs is presented in \cite{Becker16}, under the assumption that the demands and capacities of the vertices are upper bounded by a constant.

Compared to the MUC, relatively fewer special cases of the MCC problem have been studied in the literature.
We refer to the version of MCC where the underlying metric is Euclidean as the \textit{Euclidean Capacitated Covering} (ECC) problem. The dimension of the Euclidean space is assumed to be a constant. One such version arises when 
$\cb$ comprises of all possible unit
balls. This problem appeared in the Sloan Digital Sky Survey project \cite{LuptonMY98}. Building on the shifting strategy of \CCite{HochbaumM85}, \CCite{GhasemiR14} obtain a PTAS for this problem. When the set $\cb$ of balls is arbitrary, the best known approximation guarantee is $O(\log n)$, even in the plane.

Given this state of affairs for the MCC and the ECC, we focus our efforts on finding a bi-criteria approximation.
Specifically, we allow the balls in our solution to expand by at most a constant factor
$\lambda$, without changing their capacity constraints (but optimal solution does not expand). We formalize this as follows. An $(\alpha, \beta)$-approximation for a version of MCC, is a solution in which the balls may be expanded by a factor of $\beta$ (i.e. for any ball $B_i$, and any point $p_j \in P$ that is assigned to $B_i$, $d(c_i, p_j) \le \beta\cdot r_i$
), and its cost is at most $\alpha$ times that of an optimal solution (which does not expand the balls). From the reduction of Set Cover to MUC described above,
we can see that it is NP-hard to get an $(f(n),\lambda)$-approximation for any $\lambda < 3$ and $f(n)=o(\log n)$.  
We note that it is a common practice in the wireless network setting to expand the radii of antennas at the planning stage to improve the quality of service. For example, \CCite{BoseCDFKM14} propose a scheme for replacing omni-directional antennas by directional antennas that expands the antennas by a constant factor. 


\noindent {\bf Related Work.} Capacitated version of facility location and clustering type problems have been well-studied over the years. One such clustering problem is the capacitated $k$-center. In the version of this problem with
uniform capacities,  we are given a set $P$ of points in a metric space, along with an integer capacity $U$. A feasible solution to this problem is a choice of $k$ centers to open, together with an assignment of each point in $P$ to an open center, such that no center is assigned more than $U$ points. The objective is to minimize the maximum distance of a point to its assigned center.
$O(1)$ approximations are known for this problem \cite{Bar-IlanKP93,KhullerS00}; the version with non-uniform capacities is addressed in \cite{AnBCGMS15,CyganHK12}. 
Notice that the decision version of the uniform capacitated
$k$-center with the radius parameter $r$ is the same as the decision version of a special case of MCC, where the set
$\cb$ consists of balls of radius $r$ centered at each
point of the capacitated
$k$-center instance. The capacity of each ball is the same as the uniform capacity $U$ of the points. We want to find whether there is a subset of $\cb$ consisting of $k$ balls that can serve all the points without violating the capacity constraint. For capacitated versions of other related optimization problems such
as metric facility location, $k$-median etc, see~\cite{AnSS-FOCS14,LeviSS12,Li-SODA15} for recent advances.
\subsection{Our Results and Contributions.} 

In this article, we make significant progress on both the MCC and ECC problems.
\begin{compactitem}
\item \vspace{2mm} We present an $(O(1), 6.47)$-approximation for the MCC 
problem. Thus, if we are allowed to expand the input balls by a constant factor, we can obtain a solution that uses at most $O(1)$ times the number of balls used by the optimal solution. As noted above, if we are not allowed to expand by
a  factor of at least $3$, we are faced with a hardness of approximation of $\Omega(\log n)$. 
\item \vspace{2mm} We present  an $(O(\epsilon^{-4d}\log (1/\epsilon)),1+\epsilon)$-approximation for the ECC problem in $\real^d$. Thus, assuming we are allowed to expand the input balls by an arbitrarily small constant factor, we can obtain a solution with at most a corresponding constant times the number of balls used by the optimal solution. Without expansion, the best known approximation guarantee for $d \geq 3$ is $O(\log n)$, even without capacity constraints.
\end{compactitem}  

Both results are obtained via a unified scheme for rounding the natural LP relaxation for the problem. This scheme, which is instantiated in different ways to obtain the two results, may be of independent interest for obtaining similar results for related capacitated covering problems. Though the LP rounding technique is a standard tool that has been used in the literature of the capacitated problems, our actual rounding scheme is different from the existing ones. In fact, the standard rounding scheme for facility location, for example the one in \cite{LeviSS12}, is not useful for our problems, as there a point can be assigned to any facility. But in our case, each point must be assigned to a ball that contains it (modulo constant factor expansion). This hard constraint makes the covering problems more complicated to deal with.

When the input balls have the same radius, it is easier to obtain the above guarantees for the MCC and and the ECC using known results or techniques. For the MCC, this (in fact, even  a $(1, O(1))$-approximation) follows from the results for capacitated k-center \cite{Bar-IlanKP93,KhullerS00,CyganHK12,AnBCGMS15}. This is because of the connection between Capacitated $k$-center and MCC as pointed out above. 
The novelty in our work lies in handling the challenging scenario where the input balls have widely different radii. For geometric optimization problems, inputs with objects at multiple scales are often more difficult to handle than inputs with objects at the same scale.

As a byproduct of the rounding schemes we develop, the bicriteria approximations can be extended to a more general capacity model. In this model, the capacities of the balls are not necessarily the same. In particular, suppose ball $B_i$ has capacity $U_i$ and radius $r_i$. Then for any two balls $B_i, B_j\in \cb$, our model assumes the following holds: $r_i > r_j \implies U_i \ge U_j$. We refer to this capacity model as the \emph{monotonic} capacity model. We refer to the generalizations of the MCC and the ECC problems with the monotonic capacity model as the \textit{Metric Monotonic Capacitated Covering} (MMCC) problem and the \textit{Euclidean Monotonic Capacitated Covering} (EMCC) problem, respectively. We note that the monotonicity assumption on the capacities is reasonable in many applications such as wireless networks -- it might be economical to invest in capacity of an antenna to serve more clients, if it covers more area.

\noindent {\bf Hardness.} We complement our algorithmic results with some hardness of approximation results that give a better understanding of the problems we consider. Firstly, we show that for any constant $c > 1$, there exists a constant $\epsilon_c > 0$ such that it is NP-hard to obtain a $(1+\epsilon_c, c)$-approximation for the MCC problem, even when the capacity of all balls is $3$. This shows that is not possible to obtain a $(1, c)$ approximation even for an arbitrarily large constant $c$. In the hardness construction, all the balls in the hard instance do not have the same radii. This should be contrasted with the case where the radii of all balls are equal -- in this case one can use the results from capacitated $k$-center (such as \cite{AnBCGMS15,CyganHK12}), to obtain a $(1, O(1))$-approximation. 

It is natural to wonder if our algorithmic results can be extended to weighted versions of the problems. We derive hardness results that indicate that this is not possible. In particular, we show that for any constant $c\ge 1$, there exists a constant $c' > 0$, such that it is NP-hard to obtain a $(c'\log n, c)$-approximation for the weighted version of MMCC with a very simple weight function (constant power of original radius). 

We describe the natural LP relaxation for the MMCC problem in \Cref{sec:lp}. 
We describe a unified rounding scheme in \Cref{sec:framework}, and apply it
in two different ways to obtain the algorithmic guarantees for MMCC and EMCC. 
We present the two hardness results in \Cref{ap:apx}.

\section{LP relaxation for MMCC}\label{sec:lp}
Recall that the input for the MMCC consists of a set $P$ of points and a set $\cb$ of balls in some metric space, along with an integer capacity $U_i > 0$ for ball $B_i \in \cb$. We assume that for any two input balls $B_i, B_j\in \cb$, it holds that $r_i > r_j \implies U_i \ge U_j$. The goal is to compute a minimum cardinality subset $\cb' \subseteq \cb$ for which each point in $P$ can be 
assigned to a ball $\cb'$ containing it in such a way that no more than $U_i$ points are assigned to ball $B_i$.  
Let $d(p,q)$ denote the distance between two points $p$ and $q$ in the metric space. Let $B(c,r)$ denote the ball of radius $r$ centered at point $c$. We let $c_i$ and $r_i$ denote the center and radius of ball $B_i \in \cb$; thus, $B_i = B(c_i,r_i)$.

First we consider an integer programming formulation of MMCC. For each set 
$B_i \in \cb$, let $y_i = 1$ if the ball $B_i$ is selected
in the solution, and $0$ otherwise. Similarly, for each point $p_j \in X$ and each ball $B_i \in \cb$, let the variable $x_{ij} = 1$ if $p_j$ is assigned to $B_i$, and $x_{ij} = 0$ otherwise. We relax these integrality constraints, and state the corresponding linear program as follows:

\begin{align}
\label[eqn]{LP}
&\text{minimize}&\sum_{B_i \in \cb} y_i \nonumber \tag{MMCC-LP}
\\&\text{s.t.}& x_{ij} &\le y_i & & \forall p_j \in P,\; \forall B_i \in \cb
\label[constr]{constr:open}
\\& &\sum_{p_j \in P} x_{ij} & \le y_i \cdot U_i & & \forall B_i \in \cb \label[constr]{constr:cap}
\\& &\sum_{B_i \in \cb} x_{ij} &= 1 & & \forall p_j \in P
\label[constr]{constr:flow}
\\& &x_{ij} &= 0 &&\text{$\forall p_j \in P,\; \forall B_i \in \cb$ such that $p_j \not\in \cb_i$} \label[constr]{constr:coverage}
\\& &x_{ij} &\ge 0 &&\forall p_j \in P,\; \forall B_i \in \cb \label[constr]{constr:fractional_x}
\\& &0\le y_i &\le 1 &&\forall B_i \in \cb \label[constr]{constr:fractional_y}
\end{align}

Subsequently, we will
refer to an assignment $(x, y)$ that is feasible or infeasible with respect to \Crefrange{constr:open}{constr:fractional_y} as
just a \emph{solution}. The \textit{cost} of the LP solution $\sol = (x, y)$ (feasible or otherwise), denoted by $\cost(\sol)$, is defined as
$\sum_{B_i \in \cb} y_i$.

\section{The Algorithmic Framework}\label{sec:framework}
In this section, we describe our framework for extracting an integral solution 
from a fractional solution to the above LP. 
The framework consists of two major steps -- Preprocessing and the Main Rounding. The Main Rounding step is in turn divided into two smaller steps -- Cluster Formation and Selection of Objects. For simplicity of exposition, we first describe the framework with respect to the MMCC problem as an algorithm and analyze the approximation factor achieved by this algorithm for MMCC. Later, we show how one or more steps of this algorithm can be modified to obtain the desired results for the 
EMCC. 

\subsection{The Algorithm for the MMCC Problem}
Before we describe the algorithm we introduce some definitions and notation which will heavily be used throughout this section. 
For point $p_j \in P$ and ball $B_i \in \ball$, we refer to
$x_{ij}$ as the \emph{flow} from $B_i$ to $p_j$; if $x_{ij} >
0$, then we say that the ball $B_i$ \emph{serves} the point $p_j$. Each ball 
$B_i \in \ball$ can be imagined as a source of at most $y_i \cdot U_i$ units of flow, which it
distributes to some points in $P$.


\par We now define an important operation,
called \emph{rerouting of flow}. ``Rerouting of flow for a set $P' \subseteq
P$ of points from a set of balls $\cb'$ to a ball $B_k\notin \cb'$'' means obtaining a new
solution $(\hat{x}, \hat{y})$ from the current solution $(x, y)$ in the
following way: (a) For all points $p_j \in P'$, $\hat{x}_{kj} = x_{kj} +
\sum_{B_i \in \cb'} x_{ij}$; (b) for all points $p_j \in P'$ and balls $B_i \in
\cb'$, $\hat{x}_{ij} = 0$; (c) the other $\hat{x}_{ij}$ variables are the
same as the corresponding $x_{ij}$ variables.  The relevant $\hat{y_i}$ variables may also
be modified depending on the context where this operation is used.

Let $0 < \alpha \le \frac{1}{2}$ be a parameter to be fixed later. A ball $B_i \in \cb$
is \emph{heavy} if the corresponding $y_i = 1$, and \emph{light}, if $0 < y_i \le \alpha$.
Corresponding to a feasible LP solution $(x, y)$, let $\ch = \{B_i \in \cb \mid
y_i = 1\}$ denote the set of heavy balls, and $\cl = \{B_i \in \cb \mid 0 < y_i \le
\alpha \}$ denote the set of light balls. We emphasize that the set $\cl$ of light and $\ch$ of heavy balls are defined w.r.t. an LP solution; however, the reference to the LP solution may be omitted when it is clear from the context.

Now we move on towards the description of the algorithm. The algorithm, given a feasible fractional solution $\sol=({x},{y})$, rounds $\sol$ to a solution $\hat{\sigma} =(\hat{x},\hat{y})$ such that $\hat{y}$ is integral, and
the cost of $\hat{\sigma}$ is within a constant factor of the cost of $\sol$. 
The $\hat{x}$ variables are non-negative but may be fractional. Furthermore,
each point receives unit flow from the balls that are chosen (${y}$ values are $1$), and the amount of flow each chosen ball sends is bounded by its capacity. Notably, no point gets any non-zero amount of flow from a ball that is not chosen (${y}$ value is $0$). Moreover, for any ball $B_i$ and any $p_j\in P$, if $B_i$ serves $p_j$, then $d(c_i,p_j)$ is at most a constant times $r_i$. We expand each ball by a constant factor so that it contains all the points it serves. 

We note that in $\hat{\sigma}$ points might receive fractional amount of flow from the chosen balls. However, as the capacity of each ball is integral we can find, using a textbook argument for integrality of flow,  another solution with the same set of chosen balls, such that the new solution satisfies all the properties of $\hat{\sigma}$ and the additional property, that for each point $p$, there is a single chosen ball that sends one unit of flow to $p$ \cite{ChuzhoyN06}. Thus, choosing an optimal LP solution as the input $\sol=({x},{y})$ of the rounding algorithm yields a constant approximation for MMCC by expanding each ball by at most a constant factor.

Our LP rounding algorithm consists of two steps. The first step is a
preprocessing step where we construct a fractional LP solution
$\solbar=(\overline{x},\overline{y})$ from $\sol$, such that each ball in $\solbar$ is either
heavy or light, and for each point $p_j \in P$, the amount of flow that $p_j$
can potentially receive from the light balls is at most $\alpha$. The latter property
will be heavily exploited in the next step. The second step is the core step of
the algorithm where we round $\solbar$ to the desired integral solution. 


We note that throughout the algorithm, for any intermediate LP solution that we consider, we maintain the following two invariants: (i) Each ball $B_i$ sends at most $U_i$ units of flow to the points, and (ii) Each point receives exactly one unit of flow from the balls. With respect to a solution $\sol = (x,y)$, we define the {\em available capacity} of a ball $B_i \in \ball$, denoted $\AvCap(B_i)$, to be   $U_i - \sum_{p_j \in P} x_{ij}$. We now describe the preprocessing step.

\subsubsection{The Preprocessing Step}




\begin{lem}\label{lemma:preproc}
\par Given a feasible LP solution $\sol=(x, y)$, and a parameter $0 < \alpha \le
\frac{1}{2}$, there exists a polynomial time algorithm to obtain another LP solution
$\solbar=(\overline x, \overline y)$ that satisfies \Crefrange{constr:open}{constr:fractional_y} except \ref{constr:coverage} of \ref{LP}. Additionally, $\solbar$ satisfies the following properties.
\begin{compactenum}
   \item Any ball $B_i \in \cb$ with non-zero $\overline{y_i}$ is either heavy ($\overline{y_i} = 1$) or light ($0<
      \overline{y_i} \le \alpha$).
   \item For each point $p_j \in P$, we have that 
       \begin{align}
       \label[ineq]{ineq:cond2}
          \sum_{B_i \in \cl:
         \overline{x}_{ij} > 0} \overline{y_i} \leq  \alpha,
       \end{align}
where $\cl$ is the set of light balls with respect to $\solbar$. \label{Condition 2}
   \item For any heavy ball $B_i$, and any point $p_j \in P$ served by $B_i$,
      $d(c_i, p_j) \le 3r_i$.
   \item For any light ball $B_i$, and any point $p_j \in P$ served by $B_i$, $d(c_i, p_j) \le r_i$.
   \item $\cost(\solbar) \le \frac{1}{\alpha} \cost(\sol)$.
\end{compactenum}
\end{lem}
\begin{proof}
\par The algorithm starts off by initializing $\solbar$ to $\sol$. While there 
is a violation of~\Cref{ineq:cond2}, we perform the following steps.
\begin{enumerate}[leftmargin=*]
    \item We pick an arbitrary point $p_j \in P$, for which~\Cref{ineq:cond2} is
        not met. Let $\cl_j$ be a subset of light balls serving $p_j$ such that
        $\alpha < \sum_{B_i \in \cl_j} \overline{y}_i \le 2 \alpha$. Note that
        such a set $\cl_j$ always exists because the $\overline{y}_i$ variables
        corresponding to light balls are at most $\alpha \le \frac{1}{2}$. Let
        $B_k$ be a ball with the largest radius from the set $\cl_j$. (If there is more than one ball with the largest radius, we consider one having the largest capacity among those. Throughout the paper we follow this convention.) Since $r_k \ge r_m$ for all other balls $B_m \in \cl_j$, we have, by the \emph{monotonicity} assumption, that $U_k \ge U_m$. 
    \item We set $\overline{y}_k \gets \sum_{B_i \in \cl_j} \overline{y}_i$, and
        $\overline{y}_{m} \gets 0$ for $B_m \in \cl_j \setminus \{B_k\}$. Note that $\overline{y}_k \le 2\alpha \le 1$. Let $A
        = \{p_t \in P \mid \overline{x}_{it} > 0 \text{ for some } B_i \in \cl_j\setminus \{B_k\}\}$ be the set of ``affected''
        points. We reroute the flow for all the affected points in $A$ from $\cl_j \setminus \{B_k\}$ to the ball $B_k$. Since $U_k \ge U_{m}$ for
        all other balls $B_m \in \cl_j$, $B_k$ has enough available capacity 
        to ``satisfy'' all ``affected'' points. In $\solbar$, all other
        $\overline{x}_{ij}$ and $\overline{y_i}$ variables remain same as before.
        (\emph{Note}: Since $B_k$ had the largest radius from the set $\cl_j$,
        all the points in $A$ are within distance $3r_k$ from its center $c_k$,
        as seen using the triangle inequality. Also, since $\overline{y}_k >
        \alpha$, $B_k$ is no longer a light ball.
)
\end{enumerate}

\par Finally, for all balls $B_i$ such that $\overline{y}_i > \alpha$, we set
$\overline{y}_i = 1$, making them heavy. Thus $\cost(\solbar)$ is at most $\frac{1}{\alpha}$ times $\cost(\sol)$, and
$\solbar$ satisfies all the conditions stated in
the lemma.
\end{proof}\textbf{Remark.} As a byproduct of \Cref{lemma:preproc}, we get a
simple $(4,3)$-approximation algorithm for the \emph{soft} capacitated version of our
problem (see \Cref{ap:soft}). 

\subsubsection{The Main Rounding Step}
The main rounding step can logically be
divided into two stages. 
The first stage, \emph{Cluster Formation}, is the crucial step of the algorithm. Note that there can be many light balls in the preprocessed solution. Including all these balls in the final solution may incur a huge cost. Thus we use a careful strategy based on flow rerouting to select a small number of balls. 
The idea is to use the capacity of a selected light ball to reroute as much flow as possible from other intersecting balls. This in turn frees up some capacity at those balls. The available capacity of each heavy ball is used, when possible, to reroute \emph{all} the flow from some light ball intersecting it; this light ball is then added to a cluster centered around the heavy ball. Notably, for each cluster, the heavy ball is the only ball in it that actually serves some points, as we have rerouted flow from the other balls in the cluster to the heavy ball. 
In the second stage, referred to as \emph{Selection of Objects}, 
we select exactly one ball (in particular, a largest ball) from 
each cluster as part of the final solution, and reroute the flow from the heavy ball to this ball, and expand it by the required amount. Together these two stages ensure that we do not end up choosing many light balls.

We now describe the two stages in detail. Recall that any ball in the preprocessed solution is either heavy or light. Also
$\mathcal{L}$ denotes the set of light balls and $\mathcal{H}$ the set of heavy balls. Note that any heavy ball $B_i$ may serve a point $p_j$ which is at a distance $3r_i$ from $c_i$. We expand each heavy ball by a factor of $3$ so that $B_i$ can contain all points it serves. 

\begin{enumerate}[leftmargin=*]
\item \textbf{Cluster Formation.} 
In this stage, each light ball, will be added to either a set $\OLightBb$ (that will eventually be part of the final solution), or a cluster corresponding to some heavy ball. Till the very end of this stage, the sets of heavy and light balls remain unchanged. The set $\OLightBb$ is initialized to $\emptyset$. For each heavy ball $B_i$, we initialize the cluster of $B_i$, denoted by cluster$(B_i)$ to $\{B_i\}$. We say a ball is clustered if it is added to a cluster.

%
At any point, let $\Lambda$ denote the set consisting of each light ball that is (a) not in $\OLightBb$, and (b) not yet clustered. Throughout the algorithm we ensure that, if a point $p_j\in P$ is currently served by a ball $B_i\in \Lambda$, then the amount of flow $p_j$ receives from any ball $B_{i'}$ is the same as that in the preprocessed solution, i.e., the flow assignment of $p_j$ remains unchanged. While the set $\Lambda$ is non-empty, we perform the following steps.
\begin{enumerate}
\item While there is a heavy ball $B_i$ and a light ball $B_t \in \Lambda$ 
	such that (1) $B_t$ intersects $B_i$; and 
	(2) $\AvCap(B_i)$ is at least the flow $\sum_{p_j \in P}
	\overline{x}_{tj}$ out of $B_t$:
\\\- \quad \textbf{1.} For all the points served by $B_t$, we reroute the flow from $B_t$ to $B_i$.
\\\- \quad \textbf{2.} We add $B_t$ to cluster$(B_i)$.

After the execution of this \emph{while} loop, if the set $\Lambda$ becomes empty, we stop and proceed to the \emph{Selection of Objects} stage. Otherwise, we proceed to the following.
\item For any ball $B_j \in \Lambda$, let $\mathcal{A}_j$ denote the set of points currently being served by $B_j$. Also, for $B_j \in \Lambda$, let $k_j = \min\{U_j, |\mathcal{A}_j|\}$, i.e. $k_j$ denotes the minimum of its capacity, and the number of points that it currently serves.  We select the ball $B_t \in \Lambda$ with the maximum value of $k_j$, and add it to the set $\OLightBb$.
\item Since we have added $B_t$ to the set $\OLightBb$ that will be among the selected balls, we use the available capacity at $B_t$ to reroute flow to it. This is done based on the following three cases depending on the value of $k_t$.
\\\- \quad \textbf{1.} $k_t = |\mathcal{A}_t| \le U_t$. 
In this case, for each point $p_l$ in $B_t$ that gets served by $B_t$, we reroute the flow of $p_l$ from $\mathcal{B}\setminus \OLightBb$ to $B_t$. Note that after the rerouting,  $p_l$ is no longer being served by a ball in $\Lambda$. The rerouting increases the available capacity of other balls intersecting $B_t$. In particular, for each $B_i \in \mathcal{H}$, $\AvCap(B_i)$ increases by $\sum_{p_l:B_t \text{ serves } p_l} {\overline{x}}_{il}$. 
\\\- \quad \textbf{2.} $k_t = U_t < |\mathcal{A}_t|$, but $k_t = U_t > 1$. 
Observe that the flow out ball $B_t$ is $\sum_{p_j \in \mathcal{A}_t} x_{tj} \le \alpha U_t$; thus $\AvCap(B_t) \geq (1 - \alpha) U_t = (1 - \alpha) k_t$. 
\\In this case, we select a point $p_j \in \mathcal{A}_t$ arbitrarily, and reroute the flow of $p_j$ from $\mathcal{B} \setminus \OLightBb$ to $B_t$. This will increase the available capacity of other balls in $\mathcal{B} \setminus \OLightBb$ that were serving $p_j$. Also note that $p_j$ is no longer being served by a ball in $\Lambda$. 
\\We repeat the above flow rerouting process for other points of $\mathcal{A}_t$ until we encounter a point $p_l$ such that rerouting the flow of $p_l$ from $\mathcal{B} \setminus \OLightBb$ to $B_t$ violates the capacity of $B_t$. Thus the flow assignment of $p_l$ remains unchanged. 
Note that we can reroute the flow of at least $\lfloor (1-\alpha) k_t \rfloor = \lfloor (1-\alpha) U_t \rfloor \ge 1$ points of $\mathcal{A}_t$ in this manner, since $U_t > 1$ and $\alpha \le 1/2$. 
\\\- \quad \textbf{3.} $k_t = U_t = 1 < |\mathcal{A}_t|$. 
Note that $B_t$ has used $\sum_{p_j \in \mathcal{A}_t} x_{tj} \le \alpha U_t = \alpha$ capacity. 
In this case, we pick a point $p_j \in \mathcal{A}_t$ arbitrarily, and then perform the following two steps:
\\(i). Reroute the flow of $p_j$ from $\Lambda$ to $B_t$; after this, $p_j$ is no longer being served by a ball in $\Lambda$. Note that in this step, we reroute at most $\alpha$ amount of flow. Therefore, at this point we have $\AvCap(B_t) \ge 1-2\alpha$. Let $f$ be the amount of flow $p_j$ receives from the balls in $\OLightBb$. 
\\(ii). Then we reroute $\min\{\AvCap(B_t),1-f\}$ amount of flow of $p_j$ from the set $\ch$ to $B_t$. 
\end{enumerate}

When the loop terminates, we have that each light ball is either in $\OLightBb$ or clustered. We set ${\overline{y}}_i\leftarrow 1$ for each ball $B_i\in \OLightBb$, thus making it heavy. For convenience, we also set cluster$(B_i) = \{B_i\}$ for
each $B_i\in \OLightBb$. 

\item \textbf{Selection of Objects.}
At the start of this stage, we have a collection of clusters each centered around a heavy ball, such that the light balls in each cluster intersect the heavy ball. We are going to pick exactly one ball from each cluster and add it to a set
$\clusterbb$. Let $\clusterbb = \emptyset$ initially. For each heavy ball $B_i$, we consider cluster$(B_i)$ and perform the following steps. 
\begin{enumerate}
\item  If cluster$(B_i)$ consists of only the heavy ball, we add $B_i$ to $\clusterbb$. 

\item Otherwise, let $B_j$ be a largest ball in cluster$(B_i)$. If $B_j = B_i$, then we expand it by a factor of $3$. Otherwise, $B_j$ is a light ball intersecting with $B_i$, in which case we expand it by a factor of $5$. In this case, we also reroute the flow from the heavy ball to the selected ball $B_j$. Note that since we always choose a largest ball in the cluster, its capacity is at least that of the heavy ball, because of the \emph{monotonicity} assumption. We add $B_j$ to
$\clusterbb$, and we set $\overline{y}_s \leftarrow 0$ for any other ball $B_s$ in the cluster.
\end{enumerate}

After processing the clusters, we set $\overline{y}_t \leftarrow 1$ for each ball $B_t\in \clusterbb$. Finally, we return the current set of heavy balls (i.e., $\clusterbb$) as the set of selected balls. Note that the flow out of each such ball is at most its capacity, and each point receives one unit of flow from the (possibly expanded) balls that contain it. As mentioned earlier, this can be converted into an integral flow.  
\end{enumerate}

\subsubsection{The Analysis of the Rounding Algorithm}
\label{ssec:analysis}
Let $OPT$ be the cost of an optimal solution. We establish a bound on the number of balls our algorithm outputs by bounding the size of the set $\clusterbb$. Then we conclude by showing that any input ball that is part of our solution expands by at most a constant factor to cover the points it serves. 

For notational convenience, we refer to the solution $\solbar = (\overline x, \overline y)$ at hand after preprocessing, as $\sol = (x, y)$. Now we bound the size of the set $\OLightBb$ computed during Cluster Formation. The basic idea is that each light ball added to $\OLightBb$ creates significant available capacity in the heavy balls. Furthermore, whenever there is enough available capacity, a heavy ball clusters intersecting light balls, thus preventing them from being added to $\OLightBb$. The actual argument is more intricate because we need to work with a notion of $y$-accumulation, a proxy for available capacity. The way the light balls are picked for addition to $\OLightBb$ plays a crucial role in the argument. 

Let $\chp$ (resp. $\cl_1$) be the set of heavy (resp. light) balls after preprocessing, and
$I$ be the total number of iterations in the Cluster Formation stage. Also let $L_{j}$ be
the light ball selected (i.e. added to $\OLightBb$) in iteration $j$ for $1\leq j\leq I$. Now, $L_t$ maximizes $k_j$ amongst all balls from $\Lambda$ in iteration $t$ (Recall that $k_j$ was defined as the minimum of the number of points being served by $L_j$, and its capacity). Note that $k_1 \geq k_2 \geq \cdots \geq k_I$. For any $B_i \in \chp$, denote by $F(L_t, B_i)$, the total amount of flow rerouted in iteration $t$ from $B_i$ to $L_t$ corresponding to the points $B_i$ serves. 
This is the same as the increase in $\AvCap(B_i)$ when $L_t$ is added to $\OLightBb$. Correspondingly, we define $Y(L_t, B_i)$, the ``$y$-credit contributed by $L_t$ to $B_i$'',  to be $\frac{F(L_t, B_i)}{k_t}$. 
Now, the increase in available capacity over all balls in $\heavybb$ is $F_t =
\sum_{B_i\in \heavybb} F(L_t, B_i)$. The approximation guarantee of the algorithm depends crucially on the following simple lemma, which states that in each iteration we make ``sufficiently large'' amount of flow available for the set of heavy balls.

\begin{lem}\label{obs1}
Consider a ball $B_t \in \OLightBb$ processed in the \emph{Cluster Formation} stage, step c. For $0 < \alpha \le 3/8$, $F_t \geq \frac{1}{5} k_t$.
\end{lem}
\begin{proof}
The algorithm ensures that the flow asssignment of each point in $\mathcal{A}_t$ is the same as that w.r.t. the preprocessed solution. Thus by property \ref{Condition 2} of \Cref{lemma:preproc}, each such point gets at most $\alpha$ amount of flow from the balls in $\OLightBb$. Now there are three cases corresponding to the three substeps of step c.
\begin{compactenum}
\item When $k_t = |\mathcal{A}_t| \le U_t$, it is possible to reroute the flow of all points of $\mathcal{A}_t$ from $\mathcal{B} \setminus \OLightBb$ to $B_t$. Therefore, we get that $F_t \ge (1-\alpha) k_t \ge \frac{1}{5} k_t$, since $0 < \alpha \le 3/8$.
\item When $1 < k_t = U_t < |\mathcal{A}_t|$, it is possible to reroute the flow of at least $\lfloor (1-\alpha) U_t\rfloor = \lfloor (1-\alpha) k_t\rfloor$ points of $\mathcal{A}_t$ from $\mathcal{B} \setminus \OLightBb$ to $B_t$. Therefore, we get that $F_t \ge (1-\alpha)\lfloor(1-\alpha) k_t\rfloor$. When $k_t > 1$, the previous quantity is at least $\frac{1}{5} k_t$, again by using the fact that $0 < \alpha \le 3/8$.
\item When $1 = k_t = U_t < |\mathcal{A}_t|$, 
$F_t \ge (1-2\alpha)\ge \frac{1}{5} k_t$, as $0 < \alpha \le 3/8$.
\end{compactenum} \vspace{-2mm}
\end{proof}

At any moment in the Cluster Formation stage, for any ball $B_i \in \heavybb$, define its $y$-accumulation as 
\[ \yac(B_i) = \bigg({\sum_{L_t \in \OLightBb}Y(L_t, B_i)}\bigg) - \bigg({\sum_{B_j \in \cl
\cap \clst(B_i)} y_j}\bigg).\] The idea is that $B_i$ gets $y$-credit when a light
ball is added to $\OLightBb$, and loses $y$-credit when it adds a light ball to
cluster$(B_i)$; thus, $\yac(B_i)$, a proxy for the available capacity of $B_i$,
indicates the ``remaining'' $y$-credit. The next lemma
gives a relation between
the $y$-accumulation of $B_i$ and its available capacity.

\begin{restatable}{lem}{flowaccumulation}
\label{lemma:flow-acc}
Fix a heavy ball $B_i \in \heavybb$, and an integer $1 \leq t \leq I$.  Suppose that $L_1, L_2, \cdots, L_t$ have been added to $\OLightBb$. Then $\AvCap(B_i) \ge \yac(B_i) \cdot k_t$.
\end{restatable}

\begin{proof}
The proof is by induction on $t$. For this proof, we abbreviate $\AvCap(B_i)$ by $A_i$. In the first iteration, just after adding $L_1$, $A_i \geq F(L_1,B_i) = Y(L_1, B_i) \cdot k_1 \geq \yac(B_i) \cdot k_1$. 

Assume inductively that we have added balls $L_1, \cdots, L_{t-1}$ to the set $\OLightBb$, and that just after adding $L_{t-1}$, the claim is true. That is, if $\yac(B_i)$ and $A_i$ are, respectively, the $y$-accumulation and the available capacity of $B_i$ just after adding $L_{t-1}$, then $A_i \ge \yac(B_i) \cdot k_{t-1}$.

Consider the iteration $t$. At step (a) of Cluster Formation, $B_i$ uses up some of its available capacity to add $0$ or more balls to$\clst(B_i)$, after which at step (b) we add $L_t$ to $\OLightBb$. Suppose that at step (a), one or more balls are added to $\clst(B_i)$. Let $B_j$ be the first such ball, and let $k$ and $C_1$ be the number of points $B_j$ serves and the capacity of $B_j$, respectively. Then the amount of capacity used by $B_j$ is at most $$\min\{C_1\cdot y_j,k\cdot y_j\}= \min\{C_1,k\}\cdot y_j\leq k_{t-1}\cdot y_j$$ where the last inequality follows because of the order in which we add balls to $\OLightBb$. Now, after adding $B_j$ to$\clst(B_i)$, the new $y$-accumulation becomes $\yac(B_i)' = \yac(B_i) - y_j$. As for the available capacity, $$A_i' \ge A_i - k_{t-1} \cdot y_j \ge (\yac(B_i) \cdot k_{t-1}) - k_{t-1} \cdot y_j \ge (\yac(B_i) - y_j)\cdot k_{t-1} = \yac(B_i)' \cdot k_{t-1}$$
Therefore, the claim is true after addition of the first ball $B_j$. Note that $B_i$ may add multiple balls to$\clst(B_i)$, and the preceding argument would work after each such addition.

Now consider the moment when $L_t$ is added to $\OLightBb$. Let $\yac(B_i)$ denote the $y$-accumulation just before this. Now, the new $y$-accumulation of $B_i$ becomes $\yac(B_i)' = \yac(B_i) + Y(L_t, B_i)$. If $\yac(B_i) \leq 0$, then 
the new available capacity is
$$A_i' \geq F(L_t,B_i) = Y(L_t, B_i) \cdot k_t \geq
\yac(B_i)' \cdot k_t.$$ If $\yac(B_i) > 0$, the new available capacity, using the inductive hypothesis, is
$$A_i' \ge \yac(B_i) \cdot k_{t-1} + Y(L_t, B_i) \cdot k_t \ge (\yac(B_i) + Y(L_t, B_i))\cdot k_t = \yac(B_i)' \cdot k_t$$
where, in the second inequality we use $k_t \le k_{t-1}$.
\end{proof}


Now, in the next lemma, we show that any ball $B_i \in \heavybb$ cannot have ``too-much'' $y$-accumulation at any moment during Cluster Formation.
\begin{lem}\label{lemma:y-acc}
At any moment in the Cluster Formation stage, for any ball $B_i \in \heavybb$, we have that $\yac(B_i) \le 1+\alpha$.
\end{lem}
\begin{proof}
The proof is by contradiction. Let $B_i \in \heavybb$ be the first ball that
violates the condition. As $\yac(B_i)$ increases only due to addition of a light
ball to set $\OLightBb$, suppose $L_t$ was the ball whose addition to
$\OLightBb$ resulted in the violation. 

Let $\yac(B_i)$ and $\yac(B_i)' = \yac(B_i) + Y(L_t, B_i) $ be the
$y$-accumulations of $B_i$ just before and just after the addition of $L_t$.
Because of the assumption, $\yac(B_i) \le 1 + \alpha$. So the increase in the
$y$-accumulation of $B_i$ must be because $Y(L_t, B_i) > 0$. Thus, $L_t$
intersects $B_i$. However, $Y(L_t, B_i) \le 1$ by definition. Therefore, we have
$\yac(B_i) > \alpha$.

Now, by \Cref{lemma:flow-acc}, just before addition of $L_t$, $\AvCap(B_i) \ge
\yac(B_i) \cdot k_{t-1} > \alpha \cdot k_{t-1} \ge \alpha \cdot k_t$, as $k_t \le k_{t-1}$. However, $L_t$ is a light ball, and so the total flow out of $L_t$ is at most $\alpha k_t$. Therefore, the available capacity of $B_i$ is large enough that we can add $L_t$ to$\clst(B_i)$, instead of to the set $\OLightBb$, which is a contradiction. 
\end{proof}

\begin{lem}\label{lemma:size-bound}
At the end of Cluster Formation stage, we have $|\OLightBb| \le 5 \cdot \lr{(1+\alpha) \cdot |\heavybb| + \sum_{B_j \in \cl_1} y_j},$
where $0<\alpha \le 3/8$. 
\end{lem}
\begin{proof}
At the end of Cluster Formation stage, 
	\begin{align}
	\sum_{B_i \in \heavybb} \yac(B_i) &\ge \sum_{\substack{B_i \in \heavybb\\1 \le t \le I}} Y(L_t, B_i) - \sum_{B_i \in \heavybb}\ \sum_{B_j \in\clst(B_i)}y_j \nonumber
	\\&\ge \sum_{1 \le t \le I} (F_t/k_t) - \sum_{B_j \in \cl_1}y_j \tag{ $\because F_t = \sum_{B_i \in \heavybb}F(L_t, B_i) = k_t \cdot \sum_{B_i \in \heavybb} Y(L_t, B_i)$}
	\\&\ge \frac{1}{5} \cdot |\OLightBb| - \sum_{B_j \in \cl_1} y_j \label{ineq:yacc}
	\end{align}
	
Where we used \Cref{obs1} to get the last inequality.
\\Now, adding the inequality of \Cref{lemma:y-acc} over all $B_i \in \heavybb$, we have that $\sum_{B_i \in \heavybb}\yac(B_i) \le (1+\alpha)\cdot|\heavybb|$. Combining this with (\ref{ineq:yacc}) yields the desired inequality.
\end{proof}

\begin{lem} \label{lemma:factor}
The cost of the solution returned by the algorithm is at most $21$ times the cost of an optimal solution.
\end{lem}
\begin{proof}
Let $\sol=(x,y)$ be the preprocessed LP solution. Now, the total number of balls in the solution is $|\OLightBb| + |\heavybb|$. Using \Cref{lemma:size-bound}, 
\begin{align*}
|\OLightBb| + |\heavybb| &\le 5 \cdot \lr{(1+\alpha) \cdot |\heavybb| + \sum_{B_j \in \cl_1}y_j} + |\heavybb|
\\&\le (6 + 5 \alpha) \lr{\sum_{B_j \in \heavybb} y_j + \sum_{B_j \in \cl_1} y_j} 
\\&\le (6+5\alpha)\cdot\cost(\sol)
\\&\le \lr{\frac{6+5\alpha}{\alpha}}\cdot OPT 
= 21 \cdot OPT \tag{by setting $\alpha = 3/8$}
\end{align*}
\end{proof}
\begin{restatable}{lem}{expansion}
\label{lemma6}
In the algorithm each input ball is expanded by at most a factor of $9$.
\end{restatable}

\begin{proof}
Recall that when a light ball becomes heavy in the preprocessing step, it is expanded by a factor of $3$. Therefore after the preprocessing step, any heavy ball in a solution may be an expanded or unexpanded ball. 

Now, consider the selection of the balls in the second stage. If a cluster consists of only a heavy ball, then it does not expand any further. Since it might be an expanded light ball, the total expansion factor is at most $3$.

Otherwise, for a fixed cluster, let $r_l$ and $r_h$ be the radius of the largest light ball and the heavy ball,
respectively. If $r_l \geq r_h$, then the overall expansion factor is $5$.
Otherwise, if $r_l < r_h$, then the heavy ball is chosen, and it is expanded by a factor of at most $3$. Now as the heavy ball might already be expanded by a factor of $3$ during the preprocessing step,
here the overall expansion factor is $9$. 
\end{proof}

If the capacities of all balls are equal, then one can improve the expansion factor to $6.47$ by using an alternative procedure to the \emph{Selection of Balls} stage (see \Cref{ap:lemma3}).
Lastly, from~\Cref{lemma:factor} and \Cref{lemma6}, we get the following theorem.

\begin{theorem}\label{thm:main}
There is a polynomial time $(21,9)$-approximation algorithm for the MMCC problem.
\end{theorem}

\subsection{The Algorithm for the EMCC Problem}

\paragraph*{Overview of the Algorithm.} For the EMCC problem, we can exploit the structure of $\real^d$ to restrict the expansion of the balls to at most $(1+\epsilon)$, while paying in terms of the cost of the solution. In the following, we give an overview of how to adapt the stages of the framework for obtaining this result. Note that in each iteration of the preprocessing stage for MMCC, we consider a point $p_j$ and a cluster $\cl_j$ of light balls. We select a largest ball from this set and reroute the flow of other balls in $\cl_j$ to this ball. However, to ensure that the selected ball contains all the points it serves we need to expand this ball by a factor of $3$. For the EMCC problem, for the cluster $\cl_j$, we consider the bounding hypercube whose side is at most a constant times the maximum radius of any ball from $\cl_j$, and subdivide it into multiple cells. The granularity of the cells is carefully chosen to ensure that (1) Selecting the maximum radius ball among the balls whose centers are lying in that cell, and expanding it by $(1+\epsilon)$ factor is enough for rerouting the flow from all such balls to this ball, and (2) The total number of cells is $poly(1/\epsilon)$. The Cluster Formation stage for the EMCC problem is exactly the same as that for the MMCC problem. Finally, in the Selection of Balls stage, we use similar technique as in the Preprocessing stage, however one needs to be more careful to handle some technicalities that arise. We summarize the result for the EMCC problem in the following theorem. The proof of this theorem is deferred to the \Cref{app:emcc}.

\begin{restatable}{theorem}{emcc}
\label{thm:emcc}
For any $\epsilon > 0$, there is a polynomial time $(O(\epsilon^{-4d}\log (1/\epsilon)),1+\epsilon)$-approximation algorithm for the EMCC problem in $\real^d$. Moreover, for the unit radii version there is a polynomial time $(O(\epsilon^{-2d}),1+\epsilon)$-approximation algorithm.
\end{restatable}

\bibliographystyle{abbrvnat}
\bibliography{cdc_ref}

\appendix

\section{The algorithm for the EMCC Problem} \label{app:emcc}
For convenience, we restate \Cref{thm:emcc}.

\emcc*

Now we describe the algorithm in detail. For simplicity, at first we consider the $d = 2$ case. Our algorithm takes an additional input -- a constant $\epsilon > 0$, and gives an $O(\epsilon^{-8}\log (1/\epsilon))$ approximation, where each ball in the solution may be expanded by a factor of at most $1+ \epsilon$. For the EMCC problem, the Preprocessing stage is as follows.

\begin{lem}
	\label{lemma:euclidean-preproc}
	Given a feasible LP solution $\sol = (x, y)$ corresponding to an EMCC instance in $\real^2$, and parameters $0 < \alpha \le \frac{1}{2}$, and $\epsilon > 0$, 
	there exists a polynomial time algorithm to obtain another LP solution
	$\solbar=(\overline x, \overline y)$ that satisfies \Crefrange{constr:open}{constr:fractional_y} except \ref{constr:coverage} of \ref{LP}. Additionally, $\solbar$ satisfies the following properties.
	\begin{compactenum}
		\item Any ball $B_i \in \cb$ with non-zero $\overline{y}_i$ is either heavy ($\overline{y}_i = 1$), or light $0 < \overline{y}_i \le \alpha$.
		\item For each point $p_j \in P$, we have that 
		\begin{align}
		\sum_{B_i \in \cl: \bar{x}_{ij} > 0} \overline{y}_i \le \alpha \label[ineq]{ineq:emcc} 
		\end{align}
		where $\cl$ is the set of light balls with respect to $\solbar$.
		\item For any heavy ball $B_i$, and any point $p_j \in P$ served by $B_i$, $d(c_i, p_j) \le (1+\epsilon) \cdot r_i$.
		\item For any light ball $B_i$, and any point $p_j \in P$ served by $B_i$, $d(c_i, p_j) \le r_i$.
		\item $\cost(\solbar) = O(\epsilon^{-2}\log (1/\epsilon)) \cdot \cost({\sol})$.
	\end{compactenum}
\end{lem}
\begin{proof}
	As in \Cref{lemma:preproc}, in each iteration, we pick an arbitrary point $p_j \in P$ for which the \Cref{ineq:emcc} is not met, and consider a set $\cl_j$ of light balls serving $p_j$ such that
	$\alpha < \sum_{B_i \in \cl_j} {y}_i \le 2 \alpha$. We select a subset of these balls, and for each such selected ball $B_i$, we set $y_i\leftarrow 1$. For each ball $B_t$ in $\cl_j$ which is not selected, we set $y_t\leftarrow 0$. We show that the corresponding solution satisfies the desired properties. Let $r$ be the radius of a maximum radius ball from the set $\cl_j$. Now all balls from the set $\cl_j$ contain a common point $p_j$. Thus any point that belongs to a ball with radius less than $r\epsilon/2$, is within distance $(1+\epsilon)\cdot r$ from the center of a maximum radius ball. As we are going to select such a maximum radius ball and its capacity is larger than the capacity of any ball with radius less than $r\epsilon/2$, we discard balls with radius smaller than $r \epsilon/2$. We reroute all the flow from those balls to the selected ball. 
	
	Now we divide the balls into $O(\log (1/\epsilon))$ classes such that the $i^{th}$ class contains balls of radii between $2^{i-1}r\epsilon$ and $2^{i}r\epsilon$ for $0\le i\le O(\log (1/\epsilon))$. We consider each class separately and select a subset of balls from each class. Consider the $i^{th}$ class. Note that there exists an axis-parallel square of side $2^{i+2}r\epsilon$ such that the centers of the balls in $i^{th}$ class are contained in it.  
	We subdivide this square into smaller squares, by overlaying a grid of granularity $2^{i-2}r\epsilon^2$. Note that the number of smaller squares (henceforth referred to as a \emph{cell}) in the larger square is $O(\epsilon^{-2})$. We show how to select at most one light ball from each cell. 
	
	Consider a cell from the subdivision, and let $\cl_j'$ be the balls in $i^{th}$ class (with radius at least $2^{i-1}r\epsilon$) whose centers belong to this cell. Now, we select a ball $B_m \in \cl_j'$ with the maximum radius $r_m$ from the set $\cl_j'$, and reroute the flow from the other balls to $B_m$. 
	Since the center of $B_m$ is within distance $2^{i-1}r\epsilon^2$ from the center of any ball $B_l \in \cl_j'$, all the points contained in any ball $B_l \in \cl_j$ are within distance $2^{i-1}r\epsilon^2+r_l \le \epsilon r_m + r_m = (1+\epsilon) \cdot r_m$ from the center $c_m$ of the ball $B_m$. 
	
	\par Note that the capacity $U_m$ of the ball $B_m$ is at least that of the capacity of any ball from the set $\cl_j'$, because of the \emph{monotonicity} property. Furthermore, the ball $B_m$ has enough capacity to receive all the redirected flow, since $$\sum_{p_j \in P, B_l \in \cl_j' : B_l \text{ serves } p_j} x_{lj} \le \sum_{B_l \in \cl_j'} U_l \cdot y_l \le U_m \cdot \sum_{B_l \in \cl_j} y_l \le U_m \cdot 2 \cdot \alpha \le U_m.$$
	
	As $\sum_{B_l \in \cl_j} y_l > \alpha$ and we select at most $O(\epsilon^{-2}\log (1/\epsilon))$ balls in total from $\cl_j$ the increase in cost is by at most a factor of $O(\epsilon^{-2}\log (1/\epsilon))$ by a suitable choice of $\alpha$. It is easy to verify that the other properties in the statement of the lemma are also satisfied.
\end{proof}

As mentioned before, the Cluster Formation stage for EMCC is exactly the same as the one for MMCC. Note that the Cluster Formation stage increases the cost of the solution only by a constant factor. We describe and analyze the Selection of Objects stage in the following lemma. The main idea remains similar to that of \Cref{lemma:euclidean-preproc}.

\begin{restatable}{lem}{eselection}
	\label{lemma:euclidean-selection}
	There exists a scheme for the Selection of Objects stage for the EMCC problem, such that for any $\epsilon > 0$,
	\begin{compactenum}
		\item From each cluster, we choose at most $O(\epsilon^{-6})$ balls.
		\item For any chosen ball $B(c_i, r_i)$ that serves a point $p_j \in P$, we have that $d(c_i, p_j) \le (1+\epsilon) \cdot r_i$.
		\item For any chosen ball $B_i$ with capacity $U_i$, we have that $\sum_{p_j \in P} x_{ij} \le U_i$
	\end{compactenum} 
\end{restatable}

\begin{proof}
	We show how to process each cluster $C_i$ by choosing a set of balls $\cb_i \subseteq C_i$ of size $O(\epsilon^{-6})$, such that each ball in $\cb_i$ is expanded by at most $1+\epsilon$ factor. Finally, for each point $p_j \in P$ that is served by the heavy ball $B_h \in C_i$, we reroute the flow from $B_h$ to an arbitrary ball $B' \in \cb_i$ (possibly $B_h$) such that $p_j$ is contained in $B'$. We also set $\overline{y}_l \gets 0$ for all balls $B_l \in C_i \setminus \cb_i$, and $\overline{y}_l \gets 1$ for all balls $B_l \in \cb_i$. The feasibility of this solution follows easily from the monotonicity property. Finally, we return $\bigcup_i \cb_i$ over all clusters $C_i$ as the solution. It only remains to describe how to choose the set $\cb_i$ for each cluster $C_i$.
	
	If the cluster $C_i$ contains only the heavy ball $B_h$, we set $\cb_i = \{B_h\}$. In this case, we do not need any expansion. Otherwise, let $r_h$ be the radius of the heavy ball $B_h$ at the center of the cluster $C_i$, and let $r_m$ be the maximum radius of any ball from the cluster $C_i$.
	
	\par If $r_m \le r_h \cdot \epsilon/2 $, then we expand $B_h$ by a factor of $1+\epsilon$, and set $\cb_i = \{B_h\}$.

	\par Otherwise, we consider one of the following three cases. In each case, we subdivide the enclosing square of side $4r_m$ into a grid, which is very similar to \Cref{lemma:euclidean-preproc}. Therefore, we discuss in brief the granularity of the grid and the balls that are added to the set $\cb_i$.
	\\1. $r_h < r_m \cdot \epsilon/4$. In this case, $B_m$ can expand by a factor of $1 + \epsilon$ and can cover the points covered by the balls in $C_i$ that have radius smaller than $r_m \cdot \epsilon/4$, and discard them. Then, we overlay a grid of granularity $r_m \cdot \epsilon^2 / 8$, which adds $O(\epsilon^{-4})$ balls with radius at least $r_m \cdot \epsilon/4$, to the set $\cb_i$.
	\\2. $r_m \cdot \epsilon/4 \le r_h \le r_m/c$, for some constant $c > 1$. In this case, we discard balls from $C_i$ with radii less than $r_h \cdot \epsilon/4$, and then overlay a grid of granularity $r_h \cdot \epsilon^2 / 8 \ge r_m \cdot \epsilon^3 / 32$, which adds $O(\epsilon^{-6})$ balls with radius at least $r_m \cdot \epsilon^2/(16)$, to the set $\cb_i$.
	\\3. $r_h \ge r_m / c$ for some constant $c > 1$. In this case, we discard balls from $C_i$ with radii smaller than $r_m \cdot \epsilon / (2c)$, and then overlay a grid of granularity $r_m \epsilon^2 / (4c)$, which adds $O(\epsilon^{-4})$ balls with radius at least $r_m \cdot \epsilon/(2c)$, to the set $\cb_i$.
	\par In the second and the third cases above, we also add the heavy ball $B_h$ to the set $\cb_i$, if it is not added already.
\end{proof}

We note that \Cref{lemma:euclidean-preproc}, and \Cref{lemma:euclidean-selection} can be modified to work in $\real^d$. In this case, the increase in the cost of solution become $O(\epsilon^{-d}\log (1/\epsilon))$, and $O(\epsilon^{-3d})$, respectively
(where the constants inside the Big-Oh may depend exponentially on the dimension $d$). If the radii of all balls are equal, then we can improve both the bounds to $O(\epsilon^{-d})$, since grids of granularity $O(\epsilon^{-1})$ suffice. Therefore, with suitable modifications to \Cref{lemma:euclidean-preproc}, the analysis of the Cluster Formation stage from the MMCC algorithm, and \Cref{lemma:euclidean-selection}, \Cref{thm:emcc} follows.

\section{Variants of the MMCC Problem}
In this section, we consider two variants of the MMCC problem -- the version where all capacities are equal and the soft capacitated version.


%
%
\subsection{Metric Capacitated Covering Problem}\label{ap:lemma3}
\begin{restatable}{lem}{expansionuniform}
\label{uniform-expansion}
If the capacities of all balls are equal, then there exists an alternative procedure to the \emph{Selection of Balls} stage of the algorithm for MMCC, that guarantees that any ball is expanded by at most a factor of $6.47$.
\end{restatable}

\begin{proof}
If the capacities of all balls are equal (say $U$), then we proceed in the same way until the \emph{Selection of Balls} stage. Then, we use the following scheme that guarantees a smaller expansion factor for this special case. We first describe the scheme and then analyze it.

Fix a cluster obtained after the \emph{Cluster Formation} stage. If the cluster contains only a heavy ball, then we add it to a set $\mathcal{C}$ (initialized to $\emptyset$), without expansion. 

Otherwise, let $r_l$ denote the radius of a largest ball in the cluster, and let $r_h$ be the radius of the heavy ball. Let $B_l$ and $B_h$ be the corresponding balls. We consider the following $3$ cases:
\\\- \quad $r_l \ge r_h$: In this case, let $B = B_l$. We set its new radius to be $3 r_l + 2 r_h$.
\\\- \quad $\frac{1}{\sqrt{3}} \le r_l < r_h$: Let $B = B_l$. We set its new radius to be $3r_l + 2 r_h$.
\\\- \quad $r_l < \frac{1}{\sqrt{3}} r_h$: Let $B = B_h$. We set its new radius to be $r_h + 2r_l$.

Finally, if $B \neq B_h$, then we reroute the flow from $B_h$ to $B$, set $y_h \gets 0$, and add $B$ to the set $\mathcal{C}$ respectively. Finally, we set $y_i \gets 1$ for all balls $B_i \in \mathcal{C}$, and return $\mathcal{C}$ as the solution.

To analyze the scheme, note that a heavy ball at the end of \emph{Cluster Formation} stage may have been a light ball that was expanded by a factor of $3$ in the preprocessing step. Therefore, if a cluster contains only a heavy ball, then the total expansion factor is at most $3$. Otherwise, we analyze each of the $3$ cases discussed above separately.
\\\- \quad In the first case, $3r_l + 2r_h \le 5 r_l$.
\\\- \quad In the second case, $3r_l + 2r_h \le (3 + 2 \cdot \sqrt{3}) r_l < 6.47 r_l$.
\\\- \quad In the third case, $r_h + 2r_l \le (1 + 2/\sqrt{3}) r_h$. But $B_h$ might be originally a light ball that was expanded by a factor of $3$ in the preprocessing step. Therefore, the total expansion factor is at most $3 + 2 \cdot \sqrt{3} < 6.47$.
\end{proof}

\subsection{Soft capacitated version of MMCC}\label{ap:soft}

We remind the reader that in this variant, we are allowed to open
multiple identical copies of the given ball at the same location, and each such ball has a capacity same as that of the original ball. However, we need to pay a cost of $1$ for each copy.
The LP corresponding to the soft capacitated version, is the same as \ref{LP},
except that~\Cref{constr:fractional_y} is  relaxed to simply $y_i
\ge 0$. We
solve this LP, and obtain an optimal solution $(x^*, y^*)$. Then, using the
procedure from~\Cref{lemma:preproc}, we can ensure that the flow that each point
receives from the set of \emph{non-light} balls $(\cb \setminus \cl)$ is at least
$1-\alpha$. Then, opening $\frac{1}{1-\alpha}\lceil y_i \rceil$ identical copies of
each non-light ball $B_i$ ensures that at least one demand of each point is satisfied exclusively by these balls. We now expand each of the opened balls by a factor of $3$. As $y_i \geq \alpha$ for each non-light ball $B_i$, choosing $\alpha = \frac{1}{2}$ yields a simple
$4$-approximation for this version, where
each ball is expanded by a factor of at most $3$.

\section{Hardness of Approximation}\label{ap:apx}

\subsection{Hardness of Metric Monotonic Capacitated Covering}

In this section, we consider the Metric Monotonic Capacitated Covering (MMCC) problem, and show that for any constant $c \ge 1$, there exists a constant $\epsilon_c > 0$ such that it is NP-hard to obtain a $(1+\epsilon_c, c)$-approximation for the MMCC problem. Contrast this result with the result that follows from the reduction described in the Introduction and states that it is NP-hard to obtain a $(o(\log n), c)$-approximation for the MMCC for $1 \le c < 3$ -- the following construction shows that even if we relax the expansion requirement above $3$, it is not possible to obtain a PTAS for this problem. To show this, we use a gap-preserving reduction from (a version of) the 3-Dimensional Matching problem.

Consider the Maximum Bounded 3-Dimensional Matching ($3$DM-$3$) problem (defined in \cite{Patrank94}). In this problem, we are given $3$ disjoint sets of elements $X, Y, Z$, with $|X| = |Y| = |Z| = N$, and a set of ``triples'' $T \subseteq X \times Y \times Z$, such that each element $w \in W := X \cup Y \cup Z$ appears in exactly $1$, $2$ or $3$ triples of $T$. A triple $t = (x, y, z) \in T$ is said to cover $x \in X, y \in Y, z \in Z$. The goal is to find a maximum cardinality subset $M \subseteq T$ of triples that does not agree in any coordinate. Here, the elements $U \subseteq W$ that are covered by the triples in $M$ are said to be the \emph{matched} elements. If $W = U$, then the corresponding $M$ is said to be a \emph{perfect matching}. We have the following result for the $3$DM-$3$ problem from Petrank \cite{Patrank94}.

\begin{lem}[Restatement of Theorem 4.4 from \cite{Patrank94}]
\label{lem:3dm-apx}
There exists a constant $0 < \beta < 1$, such that it is $\mathsf{NP}$-hard to distinguish between the instances of the $3$DM-$3$ problem in which a \emph{perfect matching} exists, from the instances in which at most $3 \beta N$ elements are matched.
\end{lem}

\paragraph*{Reduction from $3$DM-$3$ to MMCC}
Given an instance $I$ of $3$DM-$3$ problem, we show how to reduce it to an instance $I'$ of the MMCC problem. Recall that in the version of the MMCC problem, we are allowed to expand the balls in the input by a constant factor $c \ge 1$. 

First, we show how to construct the metric space $(P \cup C, d)$ for the MMCC instance $I'$, that is induced by the shortest path metric on the following graph $G = (P \cup C, E)$. Recall that $C$ is the set of centers, and $P$ is the set of points that need to be covered by the balls centered at centers in $C$. Before describing this graph, we construct some objects that will be useful in the description.

\begin{figure}
\includegraphics[scale=0.6]{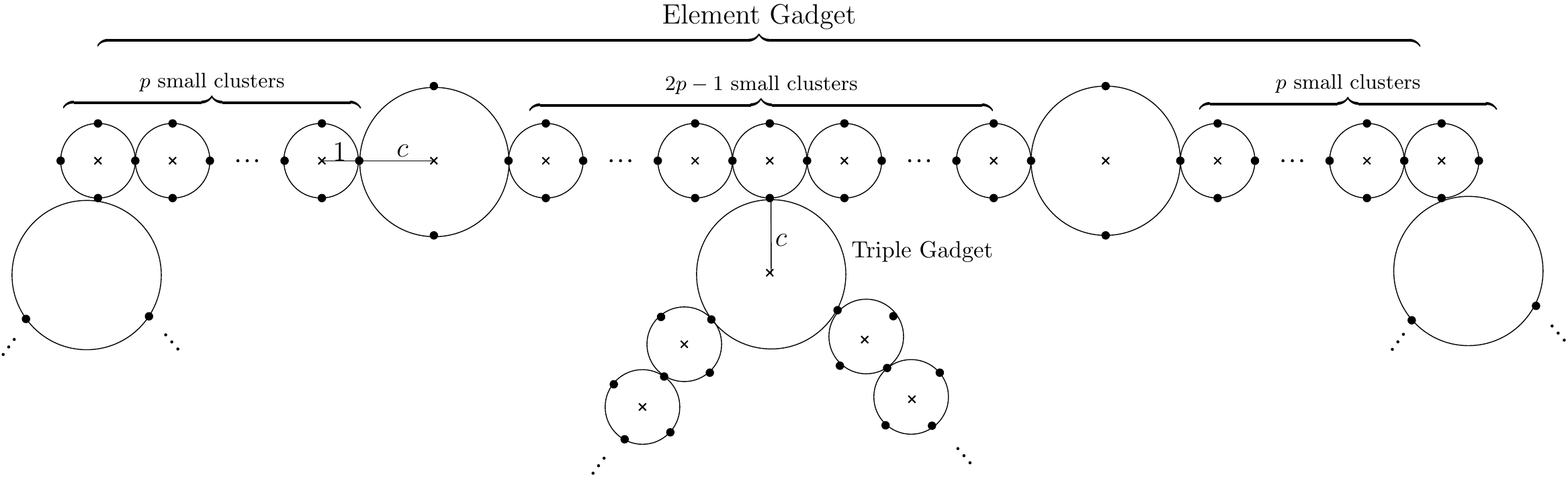}
\caption{The element $x$ belongs to three triples, whose gadgets are attached to the element gadget of $x$ at three locations. The triple gadgets are also attached to other element gadgets. Note that each ball contains $4$ points, but has the capacity of only $3$.}
\end{figure}

\par Consider a vertex $c_{1} \in C$, that is connected to $4$ other vertices $p_{1}, \cdots, p_{4} \in P$ (for convenience we refer to them as \emph{left, right, top, bottom} vertices respectively) by an edge of weight $1$. We also add a ball of radius $1$ at the center $c_1$. For convenience, we refer to this object (the $5$ vertices and the ball) as a \emph{small cluster}, and the ball as a \emph{small ball}. Similarly, if we the radius and the edge weights of the ball are $c$, then we refer to such an object as a \emph{large cluster}, and the ball as a \emph{large ball}.

\par Now, consider $p = \left\lceil\frac{c(c+1)}{2}\right\rceil + 1$ copies of \emph{small clusters}, numbered $\kappa_{1}, \cdots, \kappa_p$. For each $1 < i \le p$, we ``glue'' small clusters $\kappa_{i-1}$ and $\kappa_i$, by setting the \emph{right} vertex of $\kappa_{i-1}$ equal to the \emph{left} vertex of $\kappa_i$. This forms an object in which two consecutive clusters share exactly one vertex. We refer to this object as a \emph{small chain}. For a particular small chain, we refer to its $\kappa_1$ as its \emph{leftmost} small cluster, and to $\kappa_p$ as its \emph{rightmost} small cluster. Now, consider a \emph{big cluster} that is ``glued to'' two small chains on two sides. That is, the \emph{left} vertex of the big cluster is the same as the \emph{right} vertex of the $\kappa_p$ of a small chain (the \emph{left half}), and the \emph{right} vertex of the big cluster is the same as the \emph{left} vertex of the $\kappa_1$ of another small chain (the \emph{right half}). We call this object (which contains $2p$ small clusters and balls, and $1$ large cluster and ball) a \emph{large chain}. 

We now describe the element gadget. In an element gadget we consider two large chains $ch_1, ch_2$ that are glued together, such that the they share a common small cluster. That is, the rightmost small cluster of the right half of $ch_1$ is the same as the leftmost small cluster of the left half of $ch_2$. Denote this common cluster by $\kappa$. The respective $3$ bottom points of a) The cluster $\kappa$, b) The leftmost small cluster of the left half of $ch_1$, and c) The rightmost small cluster of the right half of $ch_2$ are referred to as \emph{ideal points}. Note that each element gadget contains $4p+1$ balls and $3(4p+1) + 1$ vertices. 

For each element $w \in W$, we add an element gadget. Now we describe the triple gadget. This gadget is similar to a \emph{large cluster}, the only difference is that in addition to the central point $c_1$, it contains only $3$ other points $p_1, p_2, p_3$. We add a triple gadget for each triple $t = (x, y, z) \in T$. Now, for each such triple $t = (x, y, z)$, we identify $p_1, p_2, p_3$ with one of the ideal points from the gadgets of $x, y, z$ respectively. Here, we ensure that if an element is contained in multiple triples, then a different ideal point is assigned to each triple. Finally, we set the capacity of each ball to be $3$.  The total number of balls in all the element gadgets is $\mathfrak{B} := 3N \cdot (4p+1)$. We refer to these balls as \emph{element balls}. Similarly, the total number of balls in all the triple gadgets is $|T|$, which we refer to as \emph{triple balls}. As mentioned above, the metric is induced by the graph $G = (P \cup C, E)$ as described above.

This completes the description of the instance $I'$ of the MMCC problem. It is easy to see that the instance $I'$ has not only monotonic capacities, but even uniform capacities. It is also worth highlighting that there are only two distinct radii in the instance $I'$. We are able to show that such a restricted version of the MMCC problem remains APX hard, even when we are allowed to expand the balls by a constant $c \ge 1$.

\begin{claim}
\label{cl:mmcc-preproc}
Consider the instance $I'$ of MMCC problem, that is obtained from an instance $I$ of $3$DM-$3$, using the above reduction. We can always convert  a solution to $I'$, where a selected ball may be expanded by a factor at most $c \ge 1$, to another feasible solution in polynomial time, where all element balls in the element gadgets are selected. Furthermore, we can ensure that in this assignment, every selected ball serves a point that is contained in it, (without any expansion). 
\end{claim}
\begin{proof}
First, notice that in each element gadget, the number of points is equal to the total capacity of all the balls in the element gadget, plus $1$. 

Now, consider a solution that does not satisfy the required property. Initially, we discard any balls in the solution if they do not use any of their capacity. Now, consider an element gadget from which a ball is not chosen in the solution. Note that this unchosen ball cannot be the central \emph{large ball} of any of its chains, because no other ball can cover the top and bottom points corresponding to it, even after expanding by a factor of $c$. We can also infer from this that, all such \emph{large balls} have at most $1$ capacity that can be used for serving other points. Without loss of generality, we assume that this remaining capacity is assigned to either left or right point, since one can always find such an assignment.

Now, any unchosen ball has to be a \emph{small ball}. Without loss of generality, we can assume that $3$ of the $4$ points that it contains, are not served by any other small ball, since we can always reassign capacities to ensure that that is not the case. Therefore, the points contained in it must be covered by triple balls, such that the corresponding triples cover the corresponding element. However, note that because of the existence of the large ball at the center of a small chain, two different triple balls cannot serve points from a common small ball. Therefore, all $3$ points must be served by a single triple ball. But any triple ball also has a capacity of $3$, and therefore, we can swap out the triple ball for this small ball, without increasing the cost of the solution. 

By repeating this process, we can include all element balls, by swapping out some of the triple balls if necessary. Now at this point, all element balls, as well as enough triple balls are included in the solution, such that the solution is feasible. Now, we assign the capacities of the selected triple balls to the corresponding \emph{ideal points}. It is easy to see that in each element gadget, at least one ideal point is served by the corresponding triple ball, and therefore, there exists a simple capacity assignment scheme to ensure that each element ball serves only the points contained in it.
\end{proof}

\begin{lem}
\label{lem:3dm-mmcc}
Consider an instance $I$ of $3$DM-$3$, and let $I'$ be the reduced MMCC instance. In the 3DM instance $I$, $M \subseteq T$ is a minimum size set of triples that covers all the elements of $W$, with $|M| = K$, if and only if the minimum cost of a solution to $I'$, wherein the balls may be expanded by up to a factor $c \ge 1$, is $\mathfrak{B} + K$.
\end{lem}
\begin{proof}
We first show that given a minimum size cover $M$ of size $K$, of the elements of $W$, how to select $\mathfrak{B} + K$ balls in the instance $I'$. Firstly, for all triples in $M$, include the corresponding balls in the triple gadget in the solution, which accounts for $K$ balls. Now, for such ball in a triple gadget, assign the capacity to serve the \emph{ideal points} of the corresponding element gadgets. Now, the number of points in an element gadget that are not served is at most the total capacity of each element gadget. Therefore, by selecting all $\mathfrak{B}$ balls, each point can be covered by one of the selected balls that it is contained in. Now we prove that there is no solution with a lesser cost.

\par Consider a minimum cost solution to $I'$, wherein the balls may be expanded by a factor at most $c \ge 1$. Then, we use \Cref{cl:mmcc-preproc} on to obtain another solution of at most the same size, in which all element balls are selected, and no ball is expanded. Now, this solution must have the same size, because the cost of an optimal solution where the balls may not be expanded, is at least the minimum cost of a solution where the balls may be expanded. Therefore, the costs of optimal solutions to the original and the relaxed version of the MMCC instance are equal for the instance $I'$.

In the new solution, each selected triple ball serves some of the ideal points. The number of selected triple balls must be at least $K$, because otherwise the set of selected triple balls corresponds to a cover of the instance $I$ of size less than $K$, which is a contradiction.
\par Now assume that the minimum cost solution to the instance $I'$, where the balls may be expanded by a factor $c \ge 1$, is $\mathfrak{B} + K$. Using the argument from the previous paragraph, we can obtain a cover of size $K$ for the $3$DM-$3$ instance $I$. To show that it is optimal, assume for contradiction, a smaller size solution, and use the argument from the first paragraph to obtain a solution to the MMCC instance $I'$ of size smaller than $\mathfrak{B} + K$, but where the balls are not expanded. Now, this contradicts the optimality of the initial solution, because as argued before, the optimal costs of the strict and relaxed versions of the MMCC problem are equal for the instance $I'$.
\end{proof}
Using this lemma, we obtain the following two corollaries, that show the gap between the instances that have a perfect matching, and those that do not have. 
\begin{cor}
\label{cor:perfect-matching}
If there exists a perfect matching in the $3$DM-$3$ instance $I$, then the corresponding MMCC instance $I'$ has an optimal solution of size exactly $\mathfrak{B} + N$.
\end{cor}

\begin{cor}
\label{cor:imperfect-matching}
If in the $3$DM-$3$ instance $I$, the maximum number of elements that can be matched is at most $3 \alpha N$ ($0 < \alpha < N$), then the minimum cost of any solution where the balls may be expanded by a factor $c \ge 1$, in the corresponding MMCC instance $I'$, is at least $\mathfrak{B} + \alpha N + \frac{3(1-\alpha)N}{2} = \left(1 + \frac{1-\alpha}{8 (3p+1)}\right) \cdot (\mathfrak{B} + N)$.
\end{cor}
\begin{proof}
In the $3$DM-$3$ instance $I$, the maximum number of elements that can be matched is at most $3\alpha N$, for some $0 < \alpha < 1$. If $M \subseteq T$ is a minimum size set of triples that covers all the elements in $W$, then we first show that $|M| \ge \alpha N + \frac{3(1-\alpha)N}{2}$.

Let $M = M_1 \cup M_2$, where $M_1$ is the maximal set of triples such that each set covers $3$ distinct elements, and $M_2$ is the remaining triples By assumption, the number of matched elements is at most $3\alpha N$, and therefore, $|M_1| \le \alpha N$. The number of elements left to be covered by $M_2$ is $3N - 3|M_1|$, and since each triple in $M_2$ covers at most $2$ new elements each, $|M_2| \ge \frac{3}{2}(N - |M_1|)$. Therefore, since $M_1$ and $M_2$ are disjoint, $|M| = |M_1| + |M_2| \ge  \alpha N + \frac{3(1-\alpha)N}{2}$.

Now, using \Cref{lem:3dm-mmcc}, we conclude that the cost of an optimal solution to $I'$ is at least $\mathfrak{B} + \alpha N + \frac{3(1-\alpha)N}{2}$. It is easy to verify that the previous quantity is exactly equal to $\left(1 + \frac{1-\alpha}{8 (3 p + 1)}\right) \cdot (\mathfrak{B} + N)$, recalling that $\mathfrak{B} = 3N \cdot (4p+1)$.
\end{proof}

Now, from \Cref{lem:3dm-apx}, \Cref{cor:perfect-matching}, and \Cref{cor:imperfect-matching}, we obtain the following APX-hardness result.

\begin{theorem}
For any constant $c \ge 1$, there exists a constant $\epsilon_c > 0$ such that it is NP-hard to obtain a $(1+\epsilon_c, c)$-approximation for the uniform capacitated version of MMCC.\footnote{$\epsilon_c$ depends inversely on $c^2$.}
\end{theorem}

\subsection{Hardness of Metric Monotonic Capacitated Covering with Weights}\label{ap:logp}
We consider a generalization of the Metric Monotonic Capacitated Covering (MMCC) problem. Like in the MMCC problem, here also we are given a set of balls $\cb$ and a set of points $P$ in a metric space. Each ball has a capacity, and the capacities of the balls are monotonic. Additionally, each ball has a non-negative real number associated with it which denotes its weight. The weight of a subset $\cb'$ of $\cb$ is the sum of the weights of the balls in $\cb'$. The goal is to find a minimum weight subset $\cb'$ of $\cb$ and compute an assignment of the points in $P$ to the balls in $\ball'$ such that the number of
points assigned to a ball is at most its capacity. We refer to this problem as Metric Monotonic Capacitated Covering with Weights (MMCC-W). In the case where all balls have the same radius and the same capacity one can get an $(1,O(1))$-approximation for MMCC-W by using a constant approximation algorithm for the Budgeted Center problem \cite{AnBCGMS15}. However, as we prove, there are instances of MMCC-W that consist of balls of only two distinct radii for which it is NP-hard to obtain a $(o(\log |P|), c)$-approximation for any constant $c$. 

The reduction is from the Set Cover problem. Recall that in Set Cover we are given a set system $(X, \cf)$ with $n = |X|$ elements
and $m = |\cf|$ subsets of $X$. For each element $e_i\in X$, let $m_i$ be the number of sets in $\cf$ that contain $e_i$. Also for each set $X_j\in \cf$, let $n_j$ be the number of elements in $X_j$. Note that $\sum_{i=1}^n m_i=\sum_{j=1}^m n_j$. Given any instance $I$ of Set Cover we construct an instance $I'=(P,\cb)$ of MMCC-W. Let $[t]=\{1,\ldots,t\}$. Fix a constant $c$, which is the factor by which the balls in the solution are allowed to be expanded, and a real $\alpha > 0$. Let $N=\max\{m,n\}$ and $M=c^{1+1/\alpha}N^{2/\alpha}$. $P$ contains $M\cdot m_i$ points corresponding to each element $e_i$. To describe the distances between the points we define a weighted graph $G$ whose vertex set is $P\cup C$, where $C$ is the set of centerpoints of the balls in $\cb$. The graph contains a set $V_i$ of $2M\cdot m_i-1$ vertices corresponding to each element $e_i$. The subgraph of $G$ induced by the vertices of $V_i$ is a path of $2M\cdot m_i-1$ vertices. We denote this path by $\pi_i$. Refer a degree $1$ vertex on $\pi_i$ as its $1^{st}$ vertex, the vertex connected to it as the $2^{nd}$ vertex and in general for $i\ge 2$, the index of the vertex connected to the $i^{th}$ vertex other than the $(i-1)^{th}$ vertex is $i+1$. The odd indexed vertices on this path belong to $P$, and the even indexed vertices belong to $C$. Thus $M\cdot m_i$ (resp. $M\cdot m_i-1$) vertices 	of the path are in $P$ (resp. $C$). The weight of each path edge is set to be $cR/M$, where $R$ is a positive real. Corresponding to each set $X_j \in \cf$, $G$ contains a vertex $u_j$ that belongs to $C$. Now for each $e_i\in X$, consider a one-to-one mapping $f$ from the set $[m_i]$ to the set of $m_i$ subsets of $X$ that contain $e_i$. For $1\le i\le m_i$, we connect the $((i-1)\cdot M+1)^{th}$ vertex of $\pi_i$ to the vertex corresponding to the set $f(i)$ by an edge of weight $R$. Note that for any set $X_j \in \cf$, $u_j$ gets connected to $n_j$ vertices of $G$. This concludes the description of $G$. 

We consider the metric space $(P\cup C,d)$ for $I'$, where $d$ is the shortest path metric on $G$. Now we describe the set of balls in $I'$. For each $X_j\in \cf$, we add the ball $B(u_j,R)$ to $\cb$ and set its capacity to $n_j$. Note that $B(u_j,R)$ contains exactly $n_j$ points of $P$. For each $e_i \in X$, now consider the set of vertices $V_i$. For each point $p$ of $C\cap V_i$, we add the ball $B(p,cR/M)$ to $\cb$ and set its capacity to $1$. We note that $B(p,cR/M)$ contains $2$ points. The balls in $\cb$ have only two distinct radii. It is not hard to see that the capacities of these balls are monotonic w.r.t their radii. We set the weight of each ball $B(p,r)$ to $r^{1+\alpha}$. 

\begin{lem}\label{lem:weighted}
The elements in $X$ can be covered by $k$ sets of $\cf$ iff there is a solution to MMCC-W for the instance $I'$ with weight at most $(k+1)\cdot R^{1+\alpha}$ where the balls in the solution can be expanded by a factor of $c$. 
\end{lem}

\begin{proof}
Let $X$ can be covered by a collection $\cf'$ of $k$ sets. We construct a feasible solution $\cb'\subseteq \cb$ to MMCC-W whose weight is at most $(k+1)\cdot R^{1+\alpha}$. For each set $X_j \in \cf'$, we add the ball $B(u_j,R)$ to $\cb'$. We assign the $n_j$ points in $B(u_j,R)$ to it. Now for each $e_i\in X$, we add balls in the following manner. Note that at least one point of $V_i\cap P$ has already been assigned to a ball in $\cb'$. Now for each point $p\in C$ on $\pi_i$, we add the ball $B(p,cR/M)$ to $\cb'$. As one point of $V_i\cap P$ is already assigned to a ball in $\cb'$, only $M\cdot m_i-1$ points of $V_i\cap P$ are left for assigning. As we have chosen $M\cdot m_i-1$ balls each of capacity $1$ corresponding to $\pi_i$ one can easily find a valid assignment of these points. Thus $\cb'$ is a feasible solution to MMCC-W. Now the weight of the balls selected w.r.t. the sets is $k\cdot R^{1+\alpha}$. The weight of the balls chosen w.r.t. each path $\pi_i$ is at most $M\cdot m_i(cR/M)^{1+\alpha}$. The total weight of the balls w.r.t. all such paths is at most $$n\cdot M\cdot m\cdot (cR/M)^{1+\alpha}\le R^{1+\alpha}$$
Thus the weight of $\cb'$ is at most $(k+1)\cdot R^{1+\alpha}$.

Now suppose there is a solution $\cb'$ to MMCC-W with weight at most $(k+1)\cdot R^{1+\alpha}$. The total number of points in $P$ is $\sum_{i=1}^n M\cdot m_i > \sum_{j=1}^m n_j$. Now the total capacities of the balls w.r.t. the sets is $\sum_{j=1}^m n_j$. Thus there must be at least one ball in $\cb'$ which is w.r.t. a path. Also the weight of the ball w.r.t. each set is $R^{1+\alpha}$. Thus there must be at most $k$ balls w.r.t. the sets in $\cf$ that are in $\cb'$. We consider the collection $\cf'\subseteq \cf$ of sets corresponding to these balls (at most $k$ in number). We claim that $\cf'$ covers all elements of $X$. Consider any element $e_i\in X$ and the path $\pi_i$ corresponding to it. Note that $\pi_i$ contains $M\cdot m_i$ points of $P$. Consider one such point $p$ and any center $p'$ on the path $\pi_j$ where $i\ne j$. Now distance between the $p'$ and $p$ is at least $2R$ and thus even after $c$ factor expansion the ball $B(p',cR/M)$ cannot contain $p$. Now consider a center point $u_t$ such that $e_i \notin X_t$. Then due to the construction the distance between $p$ and $u_t$ is at least $3R+(cR/M)\cdot 2M > cR$. Thus even after $c$ factor expansion the ball $B(u_t,R)$ cannot contain $p$. Hence $p$ must be assigned to either a ball corresponding to $\pi_i$ or a ball corresponding to a set $X_j$ that contains $e_i$. Now there are only $M\cdot m_i-1$ balls w.r.t. $\pi_i$ in $\cb$ each of whose capacity is $1$ and thus there must be a point $p \in P$ lying on $\pi_i$ that is assigned to a ball corresponding to a set $X_j$. It follows that $e_i \in X_j$. Thus $\cf'$ covers all the points of $X$.
\end{proof}

As Set Cover is NP-hard to approximate within a factor of $o(\log n)$, from Lemma \ref{lem:weighted}, we obtain the following theorem. 

\begin{theorem}
For any constant $c\ge 1$, there exists a constant $c' > 0$, such that it is NP-hard to obtain a $(c'\log |P|, c)$-approximation for MMCC-W. This result holds even for the particular weight function, where the weight of a ball is equal to a constant power of its original radius.
\end{theorem}

\end{document}